\documentclass[12pt]{article}
\usepackage{graphicx}
\usepackage{amsmath,amsthm,amssymb,amscd,cite,color,enumerate}
\theoremstyle{plain}
\newtheorem{theorem}{Theorem}[section]
\newtheorem{proposition}[theorem]{Proposition}
\newtheorem{lemma}[theorem]{Lemma}
\newtheorem{corollary}[theorem]{Corollary}

\theoremstyle{definition}

\newtheorem{remark}{Remark}[section]

\renewenvironment{proof}[1][\proofname]{    {\noindent \textit{Proof.}}
}{{\hfill{$\square$}}}

\newcommand{\rank}{{\rm rank}}
\newcommand{\hull}{{\rm Hull}}
\newcommand{\diag}{{\rm diag}}

\begin{document}
    
    \title{Hulls of Linear Codes Revisited with  Applications\thanks{This research was supported by the Thailand Research Fund and  Silpakorn University under Research Grant RSA6280042}}
    
    \author{Somphong  Jitman and Satanan Thipworawimon\thanks{ 
            S. Jitman  and S. Thipworawimon are with the  Department of Mathematics, Faculty of Science, Silpakorn University,
            Nakhon Pathom 73000, Thailand.
            Email: {sjitman@gmail.com, thipworawimon\_s@su.ac.th}.
    }}
    
    \maketitle

 \maketitle
 
 \begin{abstract} 
     Hulls of linear codes have been of interest and extensively studied  due to their rich algebraic structures and wide applications.  In this paper,  alternative characterizations of hulls of linear codes  are given  as well as their applications.   Properties of hulls of linear codes are given in terms of  their Gramians of their generator and  parity-check matrices.  Moreover, it is show  that the Gramian of a generator  matrix of every  linear code over a finite field of odd characteristic is  diagonalizable.   Subsequently, it is shown that a linear code over a finite field of odd characteristic is complementary dual if and only if it has an orthogonal basis.   Based on this characterization,      constructions of  good entanglement-assisted quantum error-correcting codes are provided.

 \end{abstract}



\section{Introduction}
\label{intro}

Hulls     have  been introduced to classify finite projective planes in \cite{AK1990}. 
Later, it turned out that the hulls of linear codes play a vital role in determining the complexity of some algorithms in coding theory  in \cite{Leon82, Leon91,  Sendrier00, Sendrier01}.  
Due to their wide applications,  hulls of linear codes  have been extensively studied. The number of linear codes of length $n$ over $\mathbb{F}_q$ whose hulls have a common dimension and the average hull dimension  of linear codes were studied in \cite{GJG2018} and \cite{S1997}.  Recently, hulls of   linear codes have been studied and  applied in constructions of  good entanglement-assisted quantum error correcting codes in \cite{GJG2018} and \cite{LC2018}.

Some families of linear codes with special hulls such as 
self-orthogonal codes and linear complementary dual (LCD)  codes have been of interest and extensively studied. Precisely, self-orthogonal codes are linear codes with maximal hull and LCD codes are linear codes of minimal hull. These codes are practically useful in communications systems, various applications, and link with other objects as shown in \cite{LN1997, P1972, EJKL2013, JLLX2010, JX2012, M1992, CG2016, J2017, CMTQ2018,CMTQ2019 , GJG2018, QZ2015, CMT2018,S2004,FKMR2019,PZL2018} and  references therein. Therefore, it is of interest to studied  hulls, families of linear codes with special hulls and their applications.

In this paper, we focus on alternative characterizations of hulls of linear codes and  their applications.  Properties of hulls of linear codes are given in terms of the  Gramians of their  generator and parity-check matrices.   Subsequently, it is shown  that the Gramian of generator    matrix of every  linear code over a finite field of odd characteristic is diagonalizable.   This implies that  a linear code over a finite field of odd characteristic is LCD if and only if it has an orthogonal basis. Some   classes of codes with special hulls such as self-orthogonal, maximal self-orthogonal, complementary dual codes are re-formalized based on these  characterizations.  Constructions of some good entanglement-assisted quantum error-correcting codes are given based on the above discussion.

The paper is organized as follows. After this introduction, the definition and  preliminary results on Euclidean hulls of linear codes are recalled in   Section~\ref{sec2}.  In Section~\ref{sec3},  characterizations  and properties of  Euclidean hulls of linear  codes are discussed as well as some  remarks on  linear codes with special Euclidean hulls.   A note on relevant results on Hermitian hulls of linear codes is given in Section~\ref{sec5}. Applications of hulls  to constructions of entanglement-assisted quantum error-correcting codes  are discussed  in Section~\ref{sec6}.

\section{Preliminaries}  \label{sec2}

Let  $\mathbb{F}_{q}$ denote the finite field of order $q$. For a positive integer $n$, a {\em linear code} of length $n$  over $\mathbb{F}_{q}$  is defined to be a  subspace of the $\mathbb{F}_{q}$-vector space $\mathbb{F}_{q}^n$. A linear code  $C$  of length $n$ over $\mathbb{F}_{q}$ is called an {\em $[n,k]_{q}$ code} if its  $\mathbb{F}_{q}$-dimension is $k$. In addition, if the minimum Hamming distance of $C$ is $d$, the code $C$  is called an  {\em $[n,k,d]_{q}$ code}.
For $\boldsymbol{u}=(u_1,u_2,\ldots,u_n)$ and $\boldsymbol{v}=(v_1,v_2,\ldots,v_n)$ in $\mathbb{F}_{q}^n$, the \textit{Euclidean inner product} of $\boldsymbol{u}$ and $\boldsymbol{v}$ is defined to be
\[\langle \boldsymbol{u}, \boldsymbol{v} \rangle:=\sum_{i=1}^{n} u_i v_i.\]
For a  linear code $C$ of length $n$ over $\mathbb{F}_{q}$,   the {\em Euclidean dual} $C^\perp$  of $C$ is defined to be  the set 
\[C^\perp=\{ \boldsymbol{v} \in  \mathbb{F}_{q}^n\mid  \langle \boldsymbol{c}, \boldsymbol{v} \rangle =0 \text{ for all }  \boldsymbol{c}\in C\}.\]
A linear code $C$ is said to be {\em Euclidean self-orthogonal} if $C\subseteq C^{\perp_E}$ and it is said to be {\em Euclidean self-dual} if $C=C^{\perp_E}$. 	A linear code  $C$ is called a {\em maximal Euclidean self-orthogonal code} if it is  Euclidean self-orthogonal and it is  not contained in any Euclidean self-orthogonal codes.  A linear code $C$ is said to be {\em Euclidean complementary dual} if  $C\cap C^\perp =\{\boldsymbol{0}\}$. The {\em Euclidean hull} of a linear  code $C$ is defined to be $\hull(C)=C\cap {C}^{\perp_E}$.   It  is not difficult to see that  a linear code $C$ is Euclidean self-orthogonal if and only if $\hull(C)=C$ and it is  Euclidean complementary dual if  and only if  $\hull(C)=\{\boldsymbol{0}\}$.

A $k\times n$ matrix $G$ over $\mathbb{F}_q$ is called a {\em generator matrix} for  an  $[n,k]_{q}$ code $C$ if the rows of $G$ form a basis for $C$. A {\em parity-check matrix} for $C$ is defined to a generator matrix of $C^\perp$.   For an $m\times n$ matrix $A$ over $\mathbb{F}_q$, by abuse of notation,  the {\em Gram matrix}  (or {\em Gramian}) of  $A$ is defined to be  $AA^T$.  The Gramian of a generator or parity-check matrix of a linear code plays an important role in the study of self-orthogonal codes, complementary dual codes, and hulls of linear codes.

\begin{proposition}[{\cite[Proposition 3.1]{GJG2018}}] \label{hull1-E}Let $C$ be a linear $[n,k,d]_q$ code with parity check matrix $H$ and generator matrix $G$. The ranks of the Gramians  $HH^{T}$ and  $GG^T$ are independent of $H$ and $G$ so that 
    $$\rank(HH^{T})=n-k-\dim (\hull(C))=n-k-\dim (\hull(C^{\perp})),$$
    and
    $$ \rank(GG^{T})=k-\dim (\hull(C))=k-\dim (\hull(C^{\perp})). $$
\end{proposition}

From this  proposition, it is well known  that  a linear   code with   generator matrix $G$  is   Euclidean self-orthogonal if and only if  the Gramian $GG^T$ is  zero  and it is Euclidean complementary dual if and only if  the Gramian $GG^T$ is non-singular. It can be summarized in the next corollary. 
\begin{corollary}Let $C$ be a linear   code with   generator matrix $G$. Then the following statements hold.
    \begin{enumerate}
        \item  $C$ is Euclidean self-orthogonal if and only if  $GG^T=[0]$. 
        \item $C$ is Euclidean complementary dual  if and only if  $GG^T$ is non-singular
    \end{enumerate}   
\end{corollary}

From Proposition \ref{hull1-E},  it is not difficult to see that  generator and parity-check matrices of linear codes can be chosen such that their Gramians are of the following special forms (cf. \cite[{Corollary 3.2}]{LP2018}). 

\begin{proposition}\label{diag1}
    Let $C$ be a linear $[n,k]_q$ code such that  $\dim(\hull(C))=\ell$.  Then  the following statements hold.
    \begin{enumerate}
        \item  There exist  
        a parity-check matrix $H$ of $C$ and    an invertible 
        $(n-k-\ell)\times (n-k-\ell) $ symmetric  matrix $A$ over $\mathbb{F}_q$ such that  the Gramian of $H$ is of the form 
        \[HH^T= \left[ \begin{array}{c|c}
        A &0 \\ \hline
        0 & 
        0
        \end{array}
        \right]. \]    
        \item  There exist  
        a generator  matrix $G$ of $C$ and    an invertible 
        $(k-\ell)\times (k-\ell) $  symmetric   matrix $B$ over $\mathbb{F}_q$ such that 	such that  the Gramian of $G$ is of the form 
        \[GG^T= \left[ \begin{array}{c|c}
        B &0 \\ \hline
        0 & 
        0
        \end{array}
        \right]. \]    
    \end{enumerate} 
\end{proposition}

Clearly,  the Gramians of generator and  parity-check matrices  of  linear codes  are  always symmetric.  Unlike  real symmetric matrices, a square symmetric matrix over finite fields does not need to be diagonalizable.  From Proposition \ref{diag1}, it is therefore  interesting to ask whether the Gramian of  a generator/parity-check matrix   of a linear code is diagonalizable.  Equivalently,  does a linear code have a generator matrix   whose Gramian is a diagonal matrix? 
In Proposition \ref{prop:DiagGram}, we provide a  solution to this problem for the  case where $q$ is an odd prime power. A partial solution for the case where $q$ is an even prime power is given in Proposition \ref{diag_maxC}.

\section{Euclidean Hulls of Linear Codes} \label{sec3}

In this section, properties of hulls of linear codes are discussed.  Alternative characterizations of the hull and the  hull dimension of linear codes are  given.  Conditions for  generator and parity-check matrices of linear codes  to have diagonalizable Gramians are  provided. 

\subsection{Characterizations of Euclidean Hulls of Linear Codes}
The Euclidean hull dimension of linear codes has been determined in terms of the rank of the Gramians of generator and parity-check matrices of linear  codes in \cite{GJG2018} (see  Proposition \ref{hull1-E}).  

In the following proposition,  alternative characterizations of the Euclidean hull dimension of linear codes are given.

\begin{proposition}\label{rank_eq} Let $C$ be a linear  $[n,k]_q$  code and let $\ell$ be a non-negative integer. Then the following statements are equivalent.     
    \begin{enumerate}[$1)$]
        \item  $\dim(\hull(C))=\ell$.
        \item  $\rank(GG^T)=k-\ell$  for every generator matrix $G$ of $C$. 
        \item $\rank(G_1G_2^T)=k-\ell$  for all generator matrices $G_1$  and $G_2$ of $C$. 
        \item  $\rank(HH^T)=n-k-\ell$  for every  parity-check matrix $H$ of $C$. 
        \item $\rank(H_1H_2^T)=n-k-\ell$  for all parity-check matrices $H_1$  and $H_2$ of $C$. 
    \end{enumerate}
\end{proposition}
\begin{proof}
    From Proposition \ref{hull1-E}, we have the equivalences  $1) \Leftrightarrow 2)$ and $1) \Leftrightarrow 4)$.  It remains  to prove  the equivalences  $2) \Leftrightarrow 3)$ and 
    $4) \Leftrightarrow 5)$.  Since the arguments of the proofs are similar,  only the  detailed proof of   $2) \Leftrightarrow 3)$ is provided.
    
    To prove   $2) \Rightarrow 3)$, let  $ G $, $G_1$ and $G_2$ be generator matrices of $C$ and assume that  	$\rank(GG^T)= k-\ell$.   Since  the rows of  $ G $, $G_1$ and $G_2$ are base for $C$,  there exist invertible $k\times k$ matrices   $E_1$ and $E_2$ such that $G_1 = E_1G$   and $ G_2 = E_2G$.  Consequently, we have  $G_1G^T_2=E_1G(E_2G)^T=E_1G(G^TE^T_2)=E_1(GG^T)E^T_2$. Since   $E_1$ and $E_2^T$ are invertible,   we have 
    \[\rank(G_1G_2^T)=\rank(E_1(GG^T)E^T_2)=\rank(GG^T) =k-\ell\] as desired
    
    The statement    $3) \Rightarrow 2)$  is obvious.
\end{proof}

Based on Proposition \ref{rank_eq}, we have the following  characterizations.

\begin{corollary}\label{cor:G1G2} Let $C$ be a linear  $[n,k]_q$  code and let $\ell$ be a non-negative integer. Then the following statements are equivalent. 
    
    \begin{enumerate}[$1)$]
        \item  $\dim(\hull(C))=\ell$.
        \item There exist nonzero elements $a_1,a_2,\dots,a_{k-\ell}$ in $\mathbb{F}_q$ and  generator matrices $G_1$  and $G_2$ of $C$  such that \[G_1G_2^T =\diag(a_1,a_2,\dots,a_{k-\ell}, 0,\dots,0).\]
        \item There exist nonzero elements $b_1,b_2,\dots,b_{n-k-\ell}$ in $\mathbb{F}_q$ and  parity-check matrices $H_1$  and $H_2$ of $C$  such that \[H_1H_2^T =\diag(b_1,b_2,\dots,b_{n-k-\ell}, 0,\dots,0).\]
    \end{enumerate}
    By convention, the set $\{a_1,a_2,\dots,a_{k-\ell}\} $ (resp., $\{b_1,b_2,\dots,b_{n-k-\ell}\} $)  will be referred to the empty set if $k-\ell=0$ (resp., $n-k-\ell=0$).
\end{corollary}

\begin{proof}
    To prove $1) \Leftrightarrow 2)$, assume that  $\dim(\hull(C))=\ell$.  Let  $G$  be a  generator matrix of $C$. 
    By Proposition \ref{rank_eq}, we have   that  $\rank(GG^T)=k-\ell$.  Applying suitable elementary row and column operations, it follows that   \[(PG)(QG)^T=PGG^TQ^{T}=\diag(a_1,a_2,\dots,a_{k-\ell}, 0,\dots,0)\]   for some nonzero elements  $a_1,a_2,\dots,a_{k-\ell}$ in $\mathbb{F}_q$ and invertible $k\times k$ matrices $P$ and $Q$ over  $\mathbb{F}_q$. 
    Let $G_1=PG$ and $G_2=QG$.  Then $G_1$ and $G_2$ are  generator matrices of $C$ such that  $G_1G_2^T =\diag(a_1,a_2,\dots,a_{k-\ell}, 0,\dots,0)$.
    
    Conversely, assume that $2)$ holds.  Then $\rank(G_1G_2^T)=k-\ell$ and hence $\dim(\hull(C))=\ell$    by Proposition \ref{rank_eq}.
    
    Since $\hull(C)=\hull(C^\perp)$, the equivalence  $1) \Leftrightarrow 3)$  can be obtained similarly.
\end{proof}

\subsection{Diagonalizability  of Gramians}

From Subsection \ref{sec3}.1,  it guarantees that for a given linear code $C$ over $\mathbb{F}_q$, there exist generator matrices $G_1$ and $G_2$ of $C$ such that $G_1G_2^T$ is a diagonal matrix. Here, we focus on the diagonalizability the Gramian of a generator matrix of a linear code. The results are given in two cases based on the characteristic of the underlying finite field.

\subsubsection{Odd Characteristics}
For an odd prime power $q$, the Gramian of a generator/parity-check  matrix of a linear code over $\mathbb{F}_q$  will be shown to be  diagonalizable. 

We begin with the following useful lemma. 
\begin{lemma}\label{lem:vv} Let $C$ be a   linear code  of length $n$ over $\mathbb{F}_q$. If $q$ is odd and $C$ is not Euclidean self-orthogonal, then there exists an element $\boldsymbol{v}\in  C$ such that $\langle \boldsymbol{v},\boldsymbol{v}\rangle\ne 0$.  In this case, $\boldsymbol{v} \notin \hull(C)$. 
\end{lemma}
\begin{proof} Assume that  $q$ is an odd prime power and $C$ is not Euclidean self-orthogonal.   Then there exist  $\boldsymbol{u}$ and $\boldsymbol{w}$  in $ C$ such that $\langle \boldsymbol{u},\boldsymbol{w}\rangle\ne 0$. If $\langle \boldsymbol{u},\boldsymbol{u}\rangle\ne 0$ or $\langle \boldsymbol{w},\boldsymbol{w}\rangle\ne 0$, we are done. Assume that $\langle \boldsymbol{u},\boldsymbol{u}\rangle= 0$ and $\langle \boldsymbol{w},\boldsymbol{w}\rangle= 0$. Let $  \boldsymbol{v}=\boldsymbol{u}+\boldsymbol{w}$. Since $q$ is odd, we have  $\langle \boldsymbol{v},\boldsymbol{v}\rangle=\langle  \boldsymbol{u},\boldsymbol{u}\rangle+ 2  \langle \boldsymbol{u},\boldsymbol{w}\rangle +\langle \boldsymbol{w},\boldsymbol{w}\rangle =2  \langle \boldsymbol{u},\boldsymbol{w}\rangle\ne 0$ as desired.  Clearly, the said element is not in $\hull(C)$. 
\end{proof}

\begin{proposition} \label{prop:DiagGram} Let $C$ be a non-zero  linear code  of length $n$ over $\mathbb{F}_q$. If $q$ is odd, then the Gramian of a generator matrix of  $C$  is diagonalizable. 
\end{proposition}
\begin{proof}  Assume that $q$ is an odd prime power. We prove by induction on the dimension of $C$.  If $\dim(C)=1$, then Gramian of a generator matrix of  $C$ is a $1\times 1$ matrix over $\mathbb{F}_q$ which is always diagonalizable.
    
    Assume that  $\dim(C)=k$ for some positive integer $k$ and assume that the statement holds true for all linear codes of dimension $k-1$.

    If $C$ is Euclidean self-orthogonal, then   $GG^T=[0]$    is diagonalizable for all generator matrices $G$ of $C$ by Proposition \ref{rank_eq}. 
    Assume that $C$ is not Euclidean self-orthogonal. Since $q$ is odd,  there exist $\boldsymbol{v} \in C$ such that  $\langle \boldsymbol{v},\boldsymbol{v}\rangle\ne 0$ by Lemma \ref{lem:vv}.  Let $D=\{ \boldsymbol{c } \in C \mid \langle \boldsymbol{v},\boldsymbol{c}\rangle=0  \}$.  Since $\langle \boldsymbol{v},\boldsymbol{v}\rangle\ne 0$, we have $C=D\oplus \langle \boldsymbol{v}\rangle$ which implies that $\dim(D)=k-1$.  By the induction hypothesis, there exists a  generator  matrix \[G=\begin{bmatrix}
    \boldsymbol{v}_1\\
    \boldsymbol{v}_2\\
    \vdots\\
    \boldsymbol{v}_{k-1}
    \end{bmatrix} \] of $D$  whose Gramian $GG^T$ is diagonal. Since $ \{\boldsymbol{v}_1,\boldsymbol{v}_2,\dots, \boldsymbol{v}_{k-1} \}\subseteq D$,   $\langle \boldsymbol{v}_i,\boldsymbol{v}\rangle= 0$  for all $1\leq i\leq k-1$. Hence, $G'=\begin{bmatrix}
    \boldsymbol{v}\\
    G
    \end{bmatrix} $ is a generator matrix for $C$ such that the Gramian  $G'G'^T$ is a diagonal matrix.
\end{proof}

The following corollary is a direct consequence of Proposition \ref{prop:DiagGram}. Since a parity-check matrix of a linear code is a generator matrix for its dual,  the above results can be restated including  the parity-check matrix easily.  


\begin{corollary}  \label{cor:odd2}  Let $C$ be a linear $[n,k]_q$ code such that  $\dim({\hull}(C))=\ell$.  If  $q$ is odd, then  the following statements hold.
    
    \begin{enumerate}
        \item There exist   nonzero elements $a_1,a_2,\dots, a_{k-\ell}$ in  $\mathbb{F}_q$  and 
        a generator matrix $G$ of $C$  such that 
        \[GG^T=  \diag(a_1,a_2,\dots,a_{k-\ell}, 0,\dots,0).
        \]  
        \item  There exist   nonzero elements $b_1,b_2,\dots, b_{n-k-\ell}$ in  $\mathbb{F}_q$  and 
        a parity-check matrix $H$ of $C$  such that 
        \[HH^T=  \diag(b_1,b_2,\dots,b_{n-k-\ell}, 0,\dots,0).
        \] 
    \end{enumerate}
\end{corollary}

Linear codes with orthogonal or orthonormal basis    are good candidates in some applications. However, in general, an orthogonal or orthonormal basis dose not need to be exist. 
The existence of an orthonormal basis  of some Euclidean complementary dual codes has been studied in \cite{CMTQ2019}. Here,    characterization for the existence of an orthogonal basis  of  Euclidean complementary dual codes over finite fields of odd characteristic can be obtained directly from  Proposition \ref{prop:DiagGram}.
\begin{corollary}  \label{cor:oddLCD}   Let $q$ be an odd prime power and let $C$ be a linear  code  over $\mathbb{F}_q$.  Then    $C$ is Euclidean complementary dual if and only if $C$ has a Euclidean orthogonal basis. 
\end{corollary}

\subsubsection{Even Characteristics}

The following results on the diagonalizability of the Gramians of  generator and  parity-check matrices of linear codes hold true for every prime powers $q$.  However, for an odd prime power  $q$, we already have stronger results  described in the previous subsection. In practice, we may assume that $q$ is a two power for the following results.

\begin{proposition}\label{diag_maxC} Let $C$ be a linear $[n,k]_q$ code such that  $\dim({\hull}(C))=\ell$.  If $\hull(C)$ is maximal  self-orthogonal in $C$, then  there exist   nonzero elements $a_1,a_2,\dots, a_{k-\ell}$ in  $\mathbb{F}_q$  and 
    a generator matrix $G$ of $C$  such that  
    \[GG^T=  \diag(a_1,a_2,\dots,a_{k-\ell}, 0,\dots,0).
    \] 
    Precisely, the Gramian of a generator matrix of a linear code $C$ whose hull is maximal  self-orthogonal in $C$ is diagonalizable. 
\end{proposition}
\begin{proof}
    Let $\mathcal{B}=\{\boldsymbol{r}_1,\boldsymbol{r}_2,\dots,\boldsymbol{r}_{\ell}\}$ be a basis of $\hull(C)$. 
    Assume that $\hull(C)$ is maximal  self-orthogonal in $C$. 
    If there exists a codeword  $\boldsymbol{x}\in C\backslash \hull(C)$ such that $\langle \boldsymbol{x}, \boldsymbol{x} \rangle=0$, then    $\langle \boldsymbol{x},\boldsymbol{c}\rangle=0$ for all $\boldsymbol{c}\in \hull(C)$. This implies that  $\hull(C)+\langle \boldsymbol{x} \rangle$ is self-orthogonal in $C$ which is containing $\hull(C)$, a contradiction.  Hence,  $\langle \boldsymbol{x}, \boldsymbol{x}\rangle\ne 0$ for all $\boldsymbol{x}\in C\backslash \hull(C)$. 
    Extending $\mathcal{B}$ to a basis  $ \mathcal{B} \cup \{\boldsymbol{t}_{\ell+1},\boldsymbol{t}_{\ell +2},\dots,\boldsymbol{t}_{k}\}$ of $C$. 
    Using the Gram-Schmidt process, $\langle \boldsymbol{t}_{\ell+1},\boldsymbol{t}_{\ell +2},\dots,\boldsymbol{t}_{k}\rangle$ contains an orthogonal  basis,  denoted by  $\{\boldsymbol{r}_{\ell+1},\boldsymbol{r}_{\ell+2},\dots,\boldsymbol{r}_{k}\}$. Hence
    $\mathcal{ B}'=\{\boldsymbol{r}_1,\boldsymbol{r}_2,\dots,\boldsymbol{r}_{\ell},$ $\boldsymbol{r}_{\ell+1},\boldsymbol{r}_{\ell+2},\dots,\boldsymbol{r}_{k}\}$  is a basis for  $C$ such that   $\langle \boldsymbol{r}_i,\boldsymbol{r}_i \rangle \ne 0 $ for all  $ \ell+1\leq i\leq k$ and 
    $\langle \boldsymbol{r}_i,\boldsymbol{r}_j \rangle = 0$   for all  $1\leq i \leq k$ and $1\leq j \leq k$   such that $i\ne j$ or $1\leq i=j\leq \ell$.
    
    For $1\leq i\leq k-\ell$, let $a_i=\langle \boldsymbol{r}_{\ell+i},\boldsymbol{r}_{\ell+i} \rangle \ne 0 $. 
    Let  $G_1=\begin{bmatrix}
    \boldsymbol{r}_{\ell+1}\\
    \vdots\\
    \boldsymbol{r}_k
    \end{bmatrix}$, $G_2=\begin{bmatrix}
    \boldsymbol{r}_{1}\\
    \vdots\\
    \boldsymbol{r}_{\ell}
    \end{bmatrix}$ and $G=\begin{bmatrix}
    G_1\\
    G_2
    \end{bmatrix}$. 
    Then   $	G_1G^T_1
    = \diag(a_1,a_2,\dots,a_{k-\ell})$,  $G_1G^T_2 = [0]$, $G_2G^T_1=[0]$
    and $G_2G^T_2=[0]$.  
    Hence, 
    \begin{align*}
    GG^T 
    = \left[ \begin{array}{c|c}
    G_1G^T_1 &G_1G^T_2 \\ \hline
    G_2G^T_1 &G_2G^T_2
    \end{array}
    \right] 
    = \left[ \begin{array}{c|c}
    \begin{array}{ccc}
    a_1&&\\ 
    &\ddots&\\
    &&a_{k-\ell}
    \end{array} &\boldsymbol{0}\\ \hline
    \boldsymbol{0} &\boldsymbol{0}
    \end{array}
    \right]
    =\diag(a_1,a_2,\dots,a_{k-\ell}, 0,\dots,0)
    \end{align*}
    as desired.
\end{proof}

Similarly to the previous proposition, we can replace a generator matrix $G$ by a parity-check matrix $H$ of $C$ and derive the  following result.

\begin{corollary}   \label{cor:ev1} Let $C$ be a linear $[n,k]_q$ code such that  $\dim({\hull}(C))=\ell$.  If ${\hull}(C)$ is maximal  self-orthogonal in $C^\perp$, then  there exist   nonzero elements $b_1,b_2,\dots, b_{n-k-\ell}$ in  $\mathbb{F}_q$  and 
    a parity-check matrix $H$ of $C$  such that 
    \[HH^T=  \diag(b_1,b_2,\dots,b_{n-k-\ell}, 0,\dots,0).
    \] 
\end{corollary}
In the case where $C$ is maximal  self-orthogonal, then $\hull(C)=C$  is maximal  self-orthogonal in $C^\perp$. Hence, we have the following corollary. 
\begin{corollary} \label{cor:maxC} Let $C$ be a linear $[n,k]_q$ code.  If $C$ is maximal  self-orthogonal, then  there exist   nonzero elements $b_1,b_2,\dots, b_{n-2k}$ in  $\mathbb{F}_q$  and 
    a parity-check matrix $H$ of $C$  whose Gramian is 
    \[HH^T=  \diag(b_1,b_2,\dots,b_{n-2k}, 0,\dots,0).
    \] 
\end{corollary}

\begin{lemma}\label{lemMaxOS}
    Let $C$ be a linear $[n,k]_q$ code such that  $\dim(\hull(C))=\ell$.    Then the following statements hold.
    \begin{enumerate}[$1)$]
        \item  If $k-\ell\leq 1$, then $\hull(C)$ is maximal self-orthogonal in $C$.
        \item If $n-k-\ell\leq 1$, then $\hull(C)$ is maximal self-orthogonal in $C^\perp$.
    \end{enumerate}
    
\end{lemma}
\begin{proof}
    To prove 1), assume that $k-\ell \leq 1$.   If $k-\ell=0$, then we have $k=\ell$ which means $\hull(C)=C$.  Hence,  $\hull(C)$ is a self-orthogonal in $C$, i.e., $C$ is maximal self-orthogonal in $C$. Assume that  $k-\ell=1$. Then there exists $\boldsymbol{v}\in C\backslash \hull(C)$. Suppose that $\langle \boldsymbol{v},\boldsymbol{v} \rangle=0$. Then $C= \langle \boldsymbol{v} \rangle+\hull(C)$. Since $\langle \boldsymbol{v},\boldsymbol{c} \rangle=0$ for all $\boldsymbol{c}\in C$,  we have $\boldsymbol{v}\in \hull (C)$ which is  a contradiction. Hence, $\langle \boldsymbol{v},\boldsymbol{v}\rangle\neq 0$. Therefore, $\hull(C)$ is maximal self-orthogonal in $C$.
    
    By replacing $C$ with  $C^\perp$ in 1), the result of  2) follows similarly. 
\end{proof}

\begin{corollary}    \label{cor:ev2} Let $C$ be a linear $[n,k]_q$ code such that  $\dim(\hull(C))=\ell$.    If $q$ is even, then   the following statements hold.
    \begin{enumerate}[$1)$]
        \item   $k-\ell\leq 1$ if and only if $\hull(C)$ is maximal self-orthogonal in $C$.
        \item $n-k-\ell\leq 1$ if and only if  $\hull(C)$ is maximal self-orthogonal in $C^\perp$.
    \end{enumerate}
    
\end{corollary}
\begin{proof} Assume that $q$ is even. The sufficient part follows from Lemma \ref{lemMaxOS}. For necessity, assume that $k-\ell >1$.
    Then there exist two linearly independent elements $\boldsymbol{v}_1$ and $\boldsymbol{v}_2$ in $C\setminus \hull(C)$.  Then  ${\langle \boldsymbol{v}_1,\boldsymbol{v}_1\rangle}\ne 0$ and ${\langle \boldsymbol{v}_2,\boldsymbol{v}_2\rangle}\ne 0$.  Since $q$ is even,  every element in $\mathbb{F}_q$ is square. Let  $a$ be an element in $\mathbb{F}_q$ such that  $a^2=\frac{\langle \boldsymbol{v}_1,\boldsymbol{v}_1\rangle}{\langle \boldsymbol{v}_2,\boldsymbol{v}_2 \rangle}$. Then 
    $\langle \boldsymbol{v}_1+a\boldsymbol{v}_2,\boldsymbol{v}_1+a\boldsymbol{v}_2 \rangle=\langle \boldsymbol{v}_1,\boldsymbol{v}_1\rangle + 2a\langle \boldsymbol{v}_1,\boldsymbol{v}_2\rangle +a^2\langle  \boldsymbol{v}_2, \boldsymbol{v}_2\rangle  = 2 \langle \boldsymbol{v}_1,\boldsymbol{v}_1\rangle =0$ and $\boldsymbol{v}_1+a\boldsymbol{v}_2\in C\setminus \hull(C)$.   Hence, $ \hull(C)+\langle \boldsymbol{v}_1+a\boldsymbol{v}_2\rangle$ is Euclidean self-orthogonal and  $\hull(C)\subsetneq \hull(C)+\langle \boldsymbol{v}_1+a\boldsymbol{v}_2\rangle \subseteq C$.  Therefore, $\hull(C)$ is not maximal self-orthogonal in $C$.
    
    The second statement follows immediately from 1).
\end{proof}

\begin{corollary}   \label{cor:ev3} Let $C$ be a non-zero  linear code  of length $n$ over $\mathbb{F}_q$. If $q$ is even and $\dim(C)-\dim(\hull(C))\leq 1$, then the Gramian of a generator matrix of  $C$  is diagonalizable. 
\end{corollary}

The diagonalizabilty studied above will be useful in the applications in Section \ref{sec6}.

\section{Hermitian Hulls of Linear Codes} \label{sec5}

For a prime power  $q$,  the \textit{Hermitian inner product} of $\boldsymbol{u}=(u_1,u_2,\ldots,u_n)$ and $\boldsymbol{v}=(v_1,v_2,\ldots,v_n)$ in $\mathbb{F}_{q^2}^n$ is defined to be  
\[\langle \boldsymbol{u}, \boldsymbol{v} \rangle_H:=\sum_{i=1}^{n} u_i {v_i^{q}}.\]
The {\em Hermitian dual}   $C^{\perp_H}$  of $C$ is defined to be the set 
\[C^{\perp_H}=\{ \boldsymbol{v} \in  \mathbb{F}_{q^2}^n\mid  \langle \boldsymbol{c}, \boldsymbol{v} \rangle_H =0 \text{ for all }  \boldsymbol{c}\in C\}.\]
The {\em Hermitian hull} of a code $C$ is $C\cap {C}^{\perp_H}$ and denote by $\hull_H(C)=C\cap C^{\perp_H}$. 
A code $C$ is said to be {\em Hermitian self-orthogonal} if $C\subseteq C^{\perp_H}$ and it is said to be {\em Hermitian complementary dual} if  $\hull_H(C)=\{\boldsymbol{0}\}$.    Clearly,  $C$ is  {Hermitian self-orthogonal} if $\hull_H(C)=C$. For an $m\times n $ matrix $A=[a_{ij}]$, denote by $A^\dagger=[a_{ji}^q]$ the conjugate transpose  of $A$. For each $\boldsymbol{v}=(v_1,v_2,\dots,v_n)\in \mathbb{F}_{q^2}^n$, denote by $\overline{\boldsymbol{v}}=(v_1^q,v_2^q,\dots , v_n^q)$ the conjugate vector of $\boldsymbol{v}$.

In this section, a discussion on Hermitian hulls of linear codes is given.  We note that most of  the results in this section can be obtained using the arguments analogous to those in Section \ref{sec3}.  Therefore,  the proofs for those results will be omitted. Some proofs are  provided if they are  required and different from those in Section \ref{sec3}.  For convenience, the theorem numbers are given in the form  \ref{sec3}.$N'$  if it corresponds to  \ref{sec3}.$N$ in Section \ref{sec3}.

The Hermitian  hull dimension of linear codes has been characterized  in \cite{GJG2018}.  Here, we provide an alternative characterizations of the Hermitian hull dimension of linear codes.

\renewcommand{\theproposition}{\ref{rank_eq}$'$}
\begin{proposition}\label{rank_eqH} Let $C$ be a linear  $[n,k]_{q^2}$  code and let $\ell$ be a non-negative integer. Then the following statements are equivalent.     
    \begin{enumerate}[$1)$]
        \item  $\dim(\hull_H(C))=\ell$.
        \item  $\rank(GG^\dagger)=k-\ell$  for every generator matrix $G$ of $C$. 
        \item $\rank(G_1G_2^\dagger)=k-\ell$  for all generator matrices $G_1$  and $G_2$ of $C$. 
        \item  $\rank(HH^\dagger)=n-k-\ell$  for every  parity-check matrix $H$ of $C$. 
        \item $\rank(H_1H_2^\dagger)=n-k-\ell$  for all parity-check matrices $H_1$  and $H_2$ of $C$. 
    \end{enumerate}
\end{proposition}

From  Proposition \ref{rank_eqH}, the following  characterizations can be obtained directly. 

\renewcommand{\thecorollary}{\ref{cor:G1G2}$'$}
\begin{corollary}\label{cor:G1G2H} Let $C$ be a linear  $[n,k]_{q^2}$  code and let $\ell$ be a non-negative integer. Then the following statements are equivalent. 
    
    \begin{enumerate}[$1)$]
        \item  $\dim(\hull_H(C))=\ell$.
        \item There exist nonzero elements $a_1,a_2,\dots,a_{k-\ell}$ in $\mathbb{F}_{q^2}$ and  generator matrices $G_1$  and $G_2$ of $C$  such that \[G_1G_2^\dagger =\diag(a_1,a_2,\dots,a_{k-\ell}, 0,\dots,0).\]
        \item There exist nonzero elements $b_1,b_2,\dots,b_{n-k-\ell}$ in $\mathbb{F}_{q^2}$ and  parity-check matrices $H_1$  and $H_2$ of $C$  such that \[H_1H_2^\dagger=\diag(b_1,b_2,\dots,b_{n-k-\ell}, 0,\dots,0).\]
    \end{enumerate} 
\end{corollary}

For an odd prime power $q$, we  show that  $GG^\dagger$  is always diagonalizable for every  generator matrix $G$ of a linear code over $\mathbb{F}_{q^2}$. 
We begin with the following useful lemma. 
\renewcommand{\thelemma}{\ref{lem:vv}$'$}
\begin{lemma}\label{lem:vvH} Let $C$ be a   linear code  of length $n$ over $\mathbb{F}_{q^2}$. If $q$ is odd and $C$ is not Hermitian self-orthogonal, then there exists an element $\boldsymbol{v}\in  C$ such that $\langle \boldsymbol{v},\boldsymbol{v}\rangle_H\ne 0$. 
\end{lemma}
\begin{proof} Assume that  $q$ is an odd prime power and $C$ is not Hermitian  self-orthogonal.   Then there exist  $\boldsymbol{u}$ and $\boldsymbol{w}$  in $ C$ such that $\langle \boldsymbol{u},\boldsymbol{w}\rangle_H\ne 0$. If $\langle \boldsymbol{u},\boldsymbol{u}\rangle_H\ne 0$ or $\langle \boldsymbol{w},\boldsymbol{w}\rangle_H\ne 0$, we are done. Assume that $\langle \boldsymbol{u},\boldsymbol{u}\rangle_H= 0$ and $\langle \boldsymbol{w},\boldsymbol{w}\rangle_H= 0$. Let $  \boldsymbol{v}=\boldsymbol{u}+\langle \boldsymbol{u},\boldsymbol{w}\rangle_H\boldsymbol{w}$. Since $q$ is odd, we have

    \begin{align*}
    \langle \boldsymbol{v},\boldsymbol{v}\rangle_H
    &= \langle  \boldsymbol{u},\boldsymbol{u}\rangle_H
    +  \langle \boldsymbol{u},\boldsymbol{w}\rangle_H ^q \langle \boldsymbol{u},\boldsymbol{w}\rangle_H 
    +  \langle \boldsymbol{u},\boldsymbol{w}\rangle_H  \langle \boldsymbol{w},\boldsymbol{u}\rangle_H 
    +\langle \boldsymbol{u},\boldsymbol{w}\rangle_H^{q+1}\langle \boldsymbol{w},\boldsymbol{w}\rangle _H\\
    &=  \langle \boldsymbol{u},\boldsymbol{w}\rangle_H ^q \langle \boldsymbol{u},\boldsymbol{w}\rangle_H 
    +  \langle \boldsymbol{u},\boldsymbol{w}\rangle_H  \langle \boldsymbol{u},\boldsymbol{w}\rangle_H ^q
    \\
    & =2  \langle \boldsymbol{u},\boldsymbol{w}\rangle_H ^q \langle \boldsymbol{u},\boldsymbol{w}\rangle_H \\
    &\ne 0\end{align*} as desired. 
\end{proof}

Applying Lemma \ref{lem:vvH} instead of Lemma \ref{lem:vv}, the next proposition can be obtained using the arguments similar to those for the proof of  Proposition \ref{prop:DiagGram}.

\renewcommand{\theproposition}{\ref{prop:DiagGram}$'$}
\begin{proposition} \label{prop:DiagGramH} Let $C$ be a non-zero  linear code  of length $n$ over $\mathbb{F}_{q^2}$. If $q$ is odd, then $GG^\dagger$ is diagonalizable for every generator generator matrix  $G$ of  $C$. 
\end{proposition}

The following corollary is a direct consequence of Proposition \ref{prop:DiagGramH}

\renewcommand{\thecorollary}{\ref{cor:odd2}$'$}
\begin{corollary}  \label{cor:odd2H}  Let $C$ be a linear $[n,k]_{q^2}$ code such that  $\dim({\hull}(C))=\ell$.  If  $q$ is odd, then  the following statements hold.
    
    \begin{enumerate}
        \item There exist   nonzero elements $a_1,a_2,\dots, a_{k-\ell}$ in  $\mathbb{F}_q$  and 
        a generator matrix $G$ of $C$  such that 
        \[GG^T=  \diag(a_1,a_2,\dots,a_{k-\ell}, 0,\dots,0).
        \]  
        \item  There exist   nonzero elements $b_1,b_2,\dots, b_{n-k-\ell}$ in  $\mathbb{F}_q$  and 
        a parity-check matrix $H$ of $C$  such that 
        \[HH^T=  \diag(b_1,b_2,\dots,b_{n-k-\ell}, 0,\dots,0).
        \] 
    \end{enumerate}
\end{corollary}

\renewcommand{\thecorollary}{\ref{cor:oddLCD}$'$}
\begin{corollary}  \label{cor:oddLCDH}   Let $q$ be an odd prime power and let $C$ be a linear  code  over $\mathbb{F}_{q^2}$.  Then    $C$ is Hermitian  complementary dual if and only if $C$ has a Hermitian orthogonal basis. 
\end{corollary}

The following results   hold true for every prime powers $q$.  However, for an odd prime power  $q$, we already have stronger results in discussion above. In practice,  we may assume that $q$ is even.

\renewcommand{\theproposition}{\ref{diag_maxC}$'$}
\begin{proposition}\label{diag_maxCH} Let $C$ be a linear $[n,k]_{q^2}$ code such that  $\dim({\hull}_H(C))=\ell$.  If $\hull_H(C)$ is maximal  self-orthogonal in $C$, then  there exist   nonzero elements $a_1,a_2,\dots, a_{k-\ell}$ in  $\mathbb{F}_{q^2}$  and 
    a generator matrix $G$ of $C$  such that 
    \[GG^\dagger=  \diag(a_1,a_2,\dots,a_{k-\ell}, 0,\dots,0).
    \]  
\end{proposition}

We can replace a generator matrix $G$ by a parity-check matrix $H$ of $C$ and derive the result as follows.

\renewcommand{\thecorollary}{\ref{cor:ev1}$'$}
\begin{corollary}  \label{cor:ev1H} Let $C$ be a linear $[n,k]_{q^2}$ code such that  $\dim({\hull}_H(C))=\ell$.  If ${\hull}_H(C)$ is maximal  self-orthogonal in $C^{\perp_H}$, then  there exist   nonzero elements $b_1,b_2,\dots, b_{n-k-\ell}$ in  $\mathbb{F}_{q^2}$  and 
    a parity-check matrix $H$ of $C$  such that 
    \[HH^\dagger=  \diag(b_1,b_2,\dots,b_{n-k-\ell}, 0,\dots,0).
    \] 
\end{corollary}

\renewcommand{\thecorollary}{\ref{cor:maxC}$'$}
\begin{corollary} \label{cor:maxCH} Let $C$ be a linear $[n,k]_{q^2}$ code.  If $C$ is maximal Hermitian self-orthogonal, then  there exist   nonzero elements $b_1,b_2,\dots, b_{n-2k}$ in  $\mathbb{F}_{q^2}$  and 
    a parity-check matrix $H$ of $C$   such that
    \[HH^\dagger=  \diag(b_1,b_2,\dots,b_{n-2k}, 0,\dots,0).
    \] 
\end{corollary}

%
%
%

\renewcommand{\thecorollary}{\ref{cor:ev2}$'$}
\begin{corollary}    \label{cor:ev2H} Let $C$ be a linear $[n,k]_{q^2}$ code such that  $\dim(\hull_H(C))=\ell$.    If $q$ is even, then   the following statements hold.
    \begin{enumerate}[$1)$]
        \item   $k-\ell\leq 1$ if and only if $\hull_H(C)$ is maximal self-orthogonal in $C$.
        \item $n-k-\ell\leq 1$ if and only if  $\hull_H(C)$ is maximal self-orthogonal in $C^{\perp_H}$.
    \end{enumerate}
    
\end{corollary}

\renewcommand{\thecorollary}{\ref{cor:ev3}$'$}
\begin{corollary}   \label{cor:ev3H} Let $C$ be a non-zero  linear code  of length $n$ over $\mathbb{F}_{q^2}$. If $q$ is even and $\dim(C)-\dim(\hull_H(C))\leq 1$, then $GG^\dagger$ is diagonalizable for every  generator matrix $G$ of  $C$.
\end{corollary}

\renewcommand{\theproposition}{\arabic{section}.\arabic{proposition}}
\renewcommand{\thelemma}{\arabic{section}.\arabic{lemma}}
\renewcommand{\thecorollary}{\arabic{section}.\arabic{corollary}}

\section{Applications} \label{sec6}

In this section,  hulls and the diagonalizability of the Gramians discussed in Sections \ref{sec3} and \ref{sec5} are  applied to constructions of  
Entanglement-Assisted Quantum Error Correcting Codes (EAQECCs).  EAQECCs were introduced in   \cite{hsieh} which can be constructed from
classical  linear codes.  In this case, the performance of the resulting quantum codes can be  determined by the performance of the underlying classical codes. Precisely, an $[[n,k,d;c]]_q$ EAQECC  encodes $k$ logical qudits into $n$ physical qudits using
$c$ copies of maximally entangled states and its  performance   is measured by its rate $\frac{k}{n}$ and net rate ($\frac{k-c}{n})$. As shown in  \cite{brun2}, the net rate of an EAQECC is positive it is possible to obtain catalytic codes.  The readers may  refer to  \cite{brun}, \cite{GJG2018} and the references therein for  more details on EAQECCs.

The following results from \cite{GJG2018} are useful for constructions of 
EAQECCs  from classical linear codes and their hulls.

\begin{proposition}[{\cite[Corollary 3.1]{GJG2018}}] 
    \label{cor:hull1}
    Let $C$ be a classical $[n,k,d]_{q}$ linear code and $C^{\bot}$ its Euclidean dual with parameters $[n,n-k,d^{\bot} ]_{q}$.
    Then there exist $[[n,k-\dim(Hull(C)), d;n-k-\dim(\hull(C))]]_{q}$ and $[[n,n-k-\dim(\hull(C)), d^{\bot};k-\dim(\hull(C))]]_{q}$ EAQECCs.
\end{proposition}

\begin{proposition}[{\cite[Corollary 3.2]{GJG2018}}] 
    \label{cor:hull}
    Let $C$ be a classical $[n,k,d]_{q^2}$ code and let $C^{\bot H}$ be
    its Hermitian dual with parameters $[n,n-k,d^{\bot H} ]_{q^2}$.
    Then there exists $[[n,k-\dim(\hull_H(C)), d;n-k-\dim(\hull_H(C))]]_{q}$ and $[[n,n-k-\dim(\hull_H(C)), d^{\bot};k-\dim(\hull_H(C))]]_{q}$ EAQECCs.
\end{proposition}

Based on the diagonalizability of  Gramians  studied in Sections \ref{sec3} and  \ref{sec5},  EAQECCs can be constructed from all linear codes over  finite fields of odd characteristic as follows.

\begin{proposition}
    \label{propCons1}
    Let $q\geq 5$ be an odd prime power and  let $C$ be a classical $[n,k,d]_{q}$ linear code  such that  $\dim(\hull(C)) =\ell$.
    Then  there exists an $[[n+r,k-\ell,d^\prime ;n-k-\ell+r]]_{q}$ EAQECC with $d\leq d^\prime \leq d+r$ for each $0\leq r\leq   k-\ell$.
\end{proposition}
\begin{proof}     If $r=0$ or $k=\ell$, then the  result follows directly from  Proposition  \ref{cor:hull1}. 
    Next,  assume that  $1\leq r\leq k-\ell$. Since $q$ is odd, there exists a generator matrix $G$ for $C$ such that the Gramian $GG^T $ is diagonalizable by Proposition \ref{prop:DiagGram}.  Precisely, 
    there exist    linearly independent elements
    $ \boldsymbol{x}_1,\boldsymbol{x}_2,\dots, \boldsymbol{x}_{k-\ell}$ in $ C$  such that $\boldsymbol{x}_i\boldsymbol{x}_i^T\ne0$ for all $1\leq i\leq n-k$ and $\boldsymbol{x}_i\boldsymbol{x}_j^T=0$ for all $1\leq i<j\leq k-\ell$.  
    
    Since $q\geq 5$, we have that $\{a^2\mid a\in \mathbb{F}_q^*\}$ contains at least $2$ elements. Hence,
    for each $i\in \{1,2,\dots,k-\ell\}$,  there exists $\alpha_i\in \mathbb{F}_{q}^*$
    such that $\alpha_i^2\ne -\boldsymbol{x}_i\boldsymbol{x}_i^T$.
    Let  $H$ be a parity check matrix for $C$ and let  $C^\prime$ be the code with parity check matrix
    \[
    H'=
    \left(\begin{array}{ccc|c}0&&&H\\ \hline
    \alpha_1&&  &\boldsymbol{x}_1\\
    &\ddots&&\vdots\\
    & &\alpha_r&\boldsymbol{x}_r
    \end{array}\right).
    \]
    Then
    \[H'(H')^T=
    \left(\begin{array}{cccc}
    HH^T&0&\dots &0\\
    0&\alpha_1^2+\boldsymbol{x}_1\boldsymbol{x}_1^T&&0\\
    \vdots  & &\ddots& \vdots \\
    0& 0&&\alpha_r^2+\boldsymbol{x}_r\boldsymbol{x}_r^T\\
    \end{array}\right).
    \]
    Since $\alpha_i^2\ne -\boldsymbol{x}_i\boldsymbol{x}_i^T$ for all $1\leq i \leq r$ and $\rank(H H^T)=  n-k-\ell$, we have that $\rank(H^\prime (H^\prime)^T)=n-k-\ell+r \geq 0$ since $\ell\leq \min\{k,n-k\}$ and $r\geq 0$. Equivalently, $\dim(\hull(C'))=\ell$.
    Since every $d-1$ columns of $H$ are linearly independent and $\alpha_i\ne 0$ for all $i\in \{1,2,\dots,r\}$,
    every $d-1$ columns of $H^\prime$ are linearly independent.
    It follows that $C'$ is an $[n+r,k, d^\prime]_{q}$ code where $d\leq d^\prime\leq d+r$.  
    Then by Proposition \ref{cor:hull1}, there exists an $[[n+r,k-\ell,d^\prime ;n-k-\ell+r]]_{q}$ EAQECC.
\end{proof}

In the same fashion, the Hermitian hulls of linear codes can be applied in constructions of EAQECCs in the following proposition.

\begin{proposition}
    \label{prop:herm1}
    Let $q\geq 3$ be  an odd prime power and  let $C$ be a classical $[n,k,d]_{q^2}$ linear code  such that  $\dim(\hull_H(C)) =\ell$.
    Then  there exists an $[[n+r,k-\ell,d^\prime ;n-k-\ell+r]]_{q}$ EAQECC with $d\leq d^\prime \leq d+r$ for each $0\leq r\leq   k-\ell$. 
\end{proposition}
\begin{proof} 
    If $r=0$ or $k=\ell$, then the  result follows directly from  Proposition  \ref{cor:hull}. Next,  assume that  $1\leq r\leq k-\ell$. Since $q$ is odd, there exists a generator matrix $G$ for $C$ such that  $GG^\dagger $ is diagonalizable by Proposition \ref{prop:DiagGramH}.  Precisely, 
    there exist    linearly independent elements
    $ \boldsymbol{x}_1,\boldsymbol{x}_2,\dots, \boldsymbol{x}_{k-\ell}$ in $ C$  such that $\boldsymbol{x}_i\boldsymbol{x}_i^\dagger\ne0$ for all $1\leq i\leq n-k$ and $\boldsymbol{x}_i\boldsymbol{x}_j^\dagger=0$ for all $1\leq i<j\leq k-\ell$.  
    
    For each $i\in \{1,2,\dots,r\}$, there exist $\alpha_i\in \mathbb{F}_{q^2}^*$ such that $\alpha_i{^{q+1}}\ne -\boldsymbol{x}_i\boldsymbol{x}_i^\dagger$ since $q\geq 3$. Let  $H$ be a generator matrix for $C^{\bot_H}$ and let  $C^\prime$ be the code with parity check matrix
    \[
    H'=
    \left(\begin{array}{ccc|c}0&&&H\\ \hline
    \alpha_1&&  &\boldsymbol{x}_1\\
    &\ddots&&\vdots\\
    & &\alpha_r&\boldsymbol{x}_r
    \end{array}\right).
    \]
    Then
    \[H'(H')^\dagger=
    \left(\begin{array}{cccc}
    HH^\dagger&0&\dots &0\\
    0&\alpha_1^{q+1}+\boldsymbol{x}_1\boldsymbol{x}_1^\dagger&&0\\
    \vdots  & &\ddots&\\
    0& 0&&\alpha_r^{q+1}+\boldsymbol{x}_r\boldsymbol{x}_r^\dagger\\
    \end{array}\right).
    \]
    Since $\alpha_i^{q+1}\ne -\boldsymbol{x}_i\boldsymbol{x}_i^\dagger$ for all $1\leq i \leq r$ and $\rank(H H^\dagger)=  n-k-\ell$, we have that $rank(H^\prime (H^\prime)^\dagger)=n-k-\ell+r\geq 0 $ since $\ell\leq \min\{k,n-k\}$ and $r\geq 0$. Equivalently, $\dim(\hull_H(C'))=\ell$.
    It is easily seen that very $d-1$ columns of $H^\prime$ are linearly independent.
    Hence, $C'$ is an $[n+r,k, d^\prime]_{q^2}$ code where $d\leq d^\prime\leq d+r$.  
    By Proposition \ref{cor:hull}, there exists an $[[n+r,k-\ell,d^\prime ;n-k-\ell+r]]_{q}$ EAQECC.
\end{proof}

Observe that     linear  $[n,k]_q$ and  $[n,k] _{q^2}$ codes   with  $\frac{n}{2}<k \leq n$  have   hull dimension $\ell\leq \min\{k,n-k\} \leq n-k$ which implies that $k-\ell  \geq   2k-n$.   From  the constructions in Propositions \ref{propCons1} and \ref{prop:herm1}, we have  an EAQECC $Q$ with parameters  $[[n+r,k-\ell,d^\prime ;n-k-\ell+r]]_{q}$  for all $0\leq r\leq   k-\ell$.  Hence,  the net rate of $Q$ is \[\frac{(k-\ell)-(n-k-\ell+r)}{n+r}=\frac{2k-n-r}{n+r}   >0\] 
for all classical linear codes with  $k> \frac{n}{2}$ and  $0\leq r<2k-n$ since $2k-n \leq k-\ell$.   
In addition, if the dimension of the input linear code is  \begin{align}
k\geq \frac{3n+r}{4},\label{eqcondition}
\end{align} its hull dimension is   $\ell\leq \min\{k,n-k\} \leq n-k \leq  n- \frac{3n+r}{4}=\frac{n-r}{4}$ which implies that $k-\ell  \geq   k-\frac{n-r}{4}\geq  \frac{3n+r}{4}-\frac{n-r}{4}= \frac{n+r}{2}$, and hence, the rate of $Q$ is \[ \frac{k-\ell}{n+r}\geq   \frac{1}{2}. \]

To obtain  EAQECCs with good minimum distances, the input linear code using Propositions \ref{propCons1} and \ref{prop:herm1}   can  be chosen from the best-known linear codes in the database of \cite{BCP1997}. Moreover, the required number of maximally entangled states $c:=n-k-\ell+r$  can be adjusted  by the parameter $r$.

\begin{remark}  We have the following observations and  suggestions   for  the constructions of   EAQECCs   in Propositions \ref{propCons1} and \ref{prop:herm1}. 
    
    \begin{enumerate}
        \item   By choosing    best-known linear codes   in  \cite{BCP1997} satisfy  the condition $ k\geq \frac{3n+r}{4}$ in \eqref{eqcondition},   all the  EAQECCs obtained in Propositions \ref{propCons1} and \ref{prop:herm1}  are good in the sense that  they have good rate and positive net rate.     Moreover, some of them have  good minimum distances.

        \item   Under the assumption   $\ell\leq k-\frac{n+r}{2}$,   EAQECCs constructed in  Propositions \ref{propCons1} and \ref{prop:herm1}   have    good rate   \[\frac{k-\ell}{n+r}\geq   \frac{1}{2} \]  and positive net rate \[\frac{(k-\ell)-(n-k-\ell+r)}{n+r}=\frac{2k-n-r}{n+r}   > 0\]  for all   $0\leq r    < 2k-n 
        $.  It is easily seen that  the condition   $\ell\leq k-\frac{n+r}{2}$  is slightly  lighter than  \eqref{eqcondition} and 
        it  is equivalent to finding  classical  linear codes with large dimension and small Euclidean/Hermitian hull dimension.
        Therefore, linear complementary dual codes studied in   \cite{M1992, CG2016,  CMTQ2018,CMTQ2019 , GJG2018,   CMT2018}  would be good candidates in constructions of   EAQECCs. 
    \end{enumerate}

\end{remark}

%
%




\end{document}